\pdfoutput=1
\documentclass[12pt,twoside,journal,draftclsnofoot,onecolumn]{IEEEtran}

\usepackage{amsmath,amsfonts,bm,amssymb,psfrag,ifthen,color,subfigure}
\usepackage{mathrsfs}
\usepackage[justification=centering,font=footnotesize]{caption}
\usepackage{todonotes}
\usepackage{comment}

\allowdisplaybreaks[4]

\usepackage[dvips]{epsfig}
\usepackage[dvips]{graphicx}
\usepackage{amsmath,amsfonts,bm,amssymb,psfrag,ifthen,color,subfigure}
\usepackage{algpseudocode}
\usepackage{algorithm}
\usepackage{algorithmicx}
\usepackage{ifthen}
\usepackage{setspace}
\usepackage{cite}
\usepackage{url}
\usepackage{longtable}
\usepackage{stfloats}  
\usepackage{graphics,booktabs,color,subfigure}
\usepackage{float}
\usepackage{diagbox}

\usepackage[dvips]{epsfig}
\usepackage[dvips]{graphicx}
\usepackage{amsmath,amsfonts,bm,amssymb,psfrag,ifthen,color,subfigure,amsthm}
\usepackage{algpseudocode}
\usepackage{algorithm}
\usepackage{algorithmicx}
\usepackage{setspace}
\usepackage{cite}
\usepackage{url}
\usepackage{longtable}
\usepackage{stfloats}  
\usepackage{graphics,booktabs,color,subfigure}
\usepackage{float}
\usepackage{diagbox}
\usepackage{multirow}

\allowdisplaybreaks[4]

\newcommand{\tabincell}[2]{\begin{tabular}{@{}#1@{}}#2\end{tabular}}

	\usepackage{graphicx}
	\newtheorem{theorem}{\bf Theorem}
	\usepackage{makecell}
	\usepackage{graphicx}
	\usepackage{todonotes} 
	
	\newtheorem{definition}{Definition}
	
\begin{document}
		\title{A Theory for Semantic Channel Coding With
			Many-to-one Source}
	
	\author{Shuai Ma, Chuanhui Zhang, Huayan Qi, Hang Li, Yue Bi,\\
		Guangming Shi,~\IEEEmembership{Fellow,~IEEE},   and Naofal Al-Dhahir,~\IEEEmembership{Fellow,~IEEE}
		\thanks{Shuai Ma is with Peng Cheng Laboratory, Shenzhen 518066, China. (e-mail:  mash01@pcl.ac.cn).}

	}
		\maketitle
		\begin{abstract}
			
			As one of the potential key technologies of 6G,    semantic communication
			is still in its infancy and there are many open problems,
			such as semantic
			entropy definition       and   semantic channel coding theory.
			To address these challenges,  we investigate      semantic information measures and   semantic channel coding theorem.
			Specifically,         we propose   a  semantic entropy definition     as the uncertainty in the semantic interpretation of random variable symbols in the context of knowledge bases,  which can be transformed into existing semantic entropy definitions under   given conditions.
			Moreover, different from traditional communications,
  semantic communications can achieve accurate transmission of semantic information under  a non-zero bit error rate.
  Based on this  property,
				we derive a semantic channel coding theorem for a typical semantic communication with many-to-one source (i.e., multiple source sequences express the same meaning), and prove its achievability and converse based on a generalized Fano's inequality.
			  Finally, numerical results  verify the effectiveness of the proposed
			semantic entropy and semantic channel coding theorem.
			
		\end{abstract}
		
		\IEEEpeerreviewmaketitle

		\section{Introduction}

		According to the classic information theory established by Claude Shannon in 1948 \cite{Shannon_ACM_2001}, communication systems  advances are made by exploring new spectrum utilization methods and new coding schemes.
		However, due to the explosive growth of   intelligent services, such as augmented reality/virtual reality, holographic communication, and autonomous driving \cite{Anthes_AC_2016,Voulodimos_CIN_2018,Akyildiz_JFET_2022}, the fifth generation communications (5G)   system is facing many bottlenecks:
		channel capacity is approaching Shannon limit \cite{Niu_MC_2021,Zhang_E_2022,Shi_CJIT_2022},   source coding efficiency is close to the Shannon information entropy/rate distortion function limit \cite{Liu_ISIT_2021}, the energy consumption   is huge \cite{Wang_WC_2020}, and  high-quality spectrum resources are  scarce \cite{Hong_SCIC_2016,Du_ITJ_2023,Cang_ITJ_2023}.
		To meet the development needs of future  sixth generation communications (6G),
		there is an urgent need for new information representation space and degrees of freedom to improve communication efficiency and transmission capacity.

		Semantic communications \cite{Huang_ITJ_2023,Zhang_SCIC_2022, Xu_IEEE_2023}   extract semantic features from raw data,  then encode and transmit semantic information,  which are expected to   alleviate the     bottlenecks faced by the current communication networks  \cite{Shi_CJIT_2022,Lan_JCIN_2021}.
		Back in 1949,  Weaver proposed a three-level communication theory  \cite{Shannon_MG_1950} as follows,
		\begin{itemize}
			\item  Technical level: How accurately can the symbols of communication
			be transmitted?
			\item   Semantic level: How precisely do the transmitted symbols convey
			the desired meaning?
			\item   Effectiveness level: How effectively does the received meaning affect
			conduct in the desired way?
		\end{itemize}
		Essentially, semantic communications reduce communication resource overhead, i.e., bandwidth, power consumption or delay,  by exploiting computing power, which focuses on the accurate transmission of semantic information.
		Furthermore, semantic communications transform traditional grammatical communications into content-oriented communications, which will be one of the key technologies of 6G   communications \cite{Shi_CM_2021,Strinati_CN_2021,Zhu_Arxiv_2023,Shi'WC23}.

		
		\subsection{Related works}
		
		Recently,  with the great progress of artificial intelligence (AI), neural networks can extract semantic information, such as images, text, and speech, which makes semantic communications feasible  \cite{Xie_JSAC_2021,Xie_ITJ_2023,Shao_JSAC_2022,Kang_2022}.
		However, there are still many open problems in semantic communications, especially lacking information metrics and theoretic guidelines to  implement and analyze semantic communications \cite{Qin_arXiv_2021,Xin_E_2022}.

		To measure the quantity of semantic information for a source, many works proposed semantic entropy  based on logical probability \cite{Carnap_RLE_1952,Floridi_MM_2004,Basu_PMC_2014},
		fuzzy mathematics theory\cite{LUCA_IAC_1972,Liu_ITFS_2020},   language understanding model\cite{Venhuizen_MDPI_2019} or  the complexity of query tasks\cite{Chattopadhyay_2020}.
		Specifically, in 1952, Carnap and Bar-Hillel \cite {Carnap_RLE_1952} used propositional logic probability instead of statistical probability to measure the semantic information contained in a sentence, that is, the higher the probability of a sentence being logically true, the less semantic information it contains.
		The amount of semantic information \cite {Floridi_MM_2004} was    represented based on the distance from the real event, which needs the ``true" semantic event as a reference.
		With background knowledge, the semantic entropy of messages is defined based on   propositional logic theory \cite {Basu_PMC_2014}.
		A semantic entropy definition was proposed based on fuzzy mathematics theory  by \cite {LUCA_IAC_1972}.
		According to the membership degree in fuzzy set theory,  the authors in \cite {Liu_ITFS_2020} defined the semantic entropy  for classification tasks, which is   used  to calculate the optimal semantic description for each category.
		In \cite {Venhuizen_MDPI_2019}, the    semantic entropy was derived based on the general structure in the language understanding model.
		Based on the complexity of query tasks, the semantic entropy was
		defined  as the minimum number of semantic queries for data  query tasks\cite {Chattopadhyay_2020}. Besides,
		a multi-grained definition of semantic information
		was proposed in \cite {Kountouris_CM_2021} for different levels of communication systems, and used R{\'e}nyi entropy \cite{Renyi_1961} to measure semantic information.
		Semantic information was  defined  as grammatical information \cite {Kolchinsky_IF_2018} describing the relationship between the system and its environment.
			In summary, the existing works \cite{Carnap_RLE_1952,Floridi_MM_2004,Basu_PMC_2014,LUCA_IAC_1972,Liu_ITFS_2020,Venhuizen_MDPI_2019,Chattopadhyay_2020} defined task-oriented semantic entropy from different perspectives, such as logical probability \cite{Carnap_RLE_1952,Floridi_MM_2004,Basu_PMC_2014}, membership degree \cite{LUCA_IAC_1972}, and world knowledge \cite{Venhuizen_MDPI_2019}.

		Besides, 	 in contrast to    traditional communications,
  semantic communications can achieve accurate transmission of semantic information with  a non-zero bit error rate. However, how to establish semantic channel capacity is still a challenge issue, and
 some works tried to prove the achievability of the semantic channel capacity.
		The ``semantic channel capacity" was  derived based on the semantic ambiguity and the logical semantic entropy of the received signal by \cite{6004632}.
		However,   a typical joint error between the wrong codeword and the received sequence \cite{Cover_book_1999}  was ignored in  \cite{6004632}   and there is no proof of the converse.
		The authors in \cite{Basak} modeled semantic communication as a Bayesian game problem, which minimized the semantic error by optimizing the sending strategy.
		In   \cite{Liu_ISIT_2021}   a semantic communication rate-distortion framework was proposed to characterize  semantic distortion and signal distortion.
		Semantic communication was modeled as a signaling game model in game theory \cite{Jinho_2022}, and described conditional mutual information based on semantic knowledge base.
		In summary, semantic channel capacity, as the rate limit    on which    a channel can accurately transmit semantic information, is still an open problem.

		\subsection{Contributions}
		
		With the above motivations, we propose a more general  definition of  semantic entropy with the help of Shannon information entropy and prove the semantic channel coding theorem for a typical semantic communications.
		Specifically,
		the main contributions of this paper are summarized as follows:	
		

		\begin{itemize}
			\item
			
			The existing semantic entropy definitions are task-oriented,  which   limits its generalizability.  To address this issue, we adopt the a more general  definition of  semantic entropy as the uncertainty in the semantic interpretation of random variable symbols in the context of knowledge bases. Our proposed    semantic entropy definition not only depends on the probability distribution, but also depends on the specific value of the symbol and the background knowledge base. Different from the traditional measure of information,  our proposed    semantic entropy has two unique features: on the one hand,  semantic information has a certain degree of ambiguity, that is, different symbols can express the same semantic meaning; on the other hand, the same symbol may have different semantic meanings under different knowledge bases. Under the given conditions, semantic entropy can be transformed into existing semantic entropy definitions.

			\item Furthermore,   we establish a semantic channel coding theorem for semantic communications with  ``many-to-one"  semantic source, i.e., multiple source sequences express the same meaning. We first prove its achievability based on the jointly typicality tools and many-to-one mapping relation.
			We then prove the converse based on a generalized Fano's inequality. To our best knowledge, it is  the first time theoretically established the fundamental limits of semantic communications with many-to-one semantic source. This provides a theoretical tool for the  semantic communication design.

			
			\item Simulation results verify the feasibility and effectiveness of our proposed semantic entropy, and prove that with a non-zero bit error rate, the accuracy of the classification task is still high, which verifies the rationality of our proposed semantic channel coding theorem.
			
		\end{itemize}

		\begin{table}[htbp]
	\caption{Summary of Key Notations}
	\label{tablepar}
	\centering
	\begin{tabular}{|m{1.8cm}<{\centering}|m{9.2cm}|}
		\hline
		\rule{0pt}{8pt}Notations  &    Meanings \\ \hline
		
		\rule{0pt}{7.5pt}$X,Y $ &  \tabincell{c}{Random variables} \\ \hline
		
		\rule{0pt}{7.5pt}$x,y $ &  Values of random variables  \\ \hline
		
		\rule{0pt}{7.5pt}${{S}}$ &  The semantic interpretation of random variable $X$\\ \hline
		
		\rule{0pt}{7.5pt}${\text{Pr}}\left(  \cdot  \right)$ &  \tabincell{c}{The probability of  ``$ \cdot $" occurs}\\ \hline	
		
		\rule{0pt}{7.5pt}${{\cal X}}$, ${{\cal Y}}$ &  \tabincell{c}{The finite sets}\\ \hline

		\rule{0pt}{7.5pt}${\left\lceil {\alpha nR} \right\rceil } $ &  \tabincell{c}{The smallest integer not less than $ {\alpha nR}$}\\ \hline

		\rule{0pt}{7.5pt} $m$  & The index of the message set with the same semantic   \\ \hline

		\rule{0pt}{7.5pt} $x^n$, $y^n$  &  \tabincell{c}{ $n$-length sequences}  \\ \hline
		\rule{0pt}{7.5pt} $X^n$, $Y^n$   &  \tabincell{c}{ $n$-length sequences of random variables}  \\ \hline
		\rule{0pt}{7.5pt}  ${{\cal X}^n}$, ${{\cal Y}^n}$  & The sequence sets of $X^n$, $Y^n$   \\ \hline
		
		\rule{0pt}{7.5pt}  $\varepsilon $ & An error message   \\ \hline
		\rule{0pt}{7.5pt} $\emptyset $  &  \tabincell{c}{Null set}  \\ \hline

		\rule{0pt}{7.5pt}$\left|  \cdot  \right|$ &  \tabincell{c}{The number of containing elements}\\ \hline
		
	\end{tabular}
\end{table}

		The rest of this paper is organized as follows.   We  propose a    definition of  semantic entropy  in Section II. We propose a semantic channel coding theorem and prove its achievability  in Section III.
		In Section IV, we derive and prove the generalized Fano's inequality and use it to  prove the converse of the semantic channel coding  theorem. Finally, Section V concludes the paper. Table I presents  the
		key notations in  this paper and their meanings.

		\section{Proposed Definition of Semantic Entropy}

		Shannon entropy is a functional
		of the distribution of a random variable, which does not depend on the actual states taken by
		the random variable, but only on their probabilities.  As a result,      Shannon entropy can not be directly applied  to measure the semantic information.
		Different from   Shannon information theory,   semantic information has   the following  two characteristics:
		\begin{itemize}	
			\item  The semantic  has a certain degree of ambiguity. That is, different symbols can express the same meaning.  For example, ``red" and ``crimson" both mean ``red" in the given knowledge base condition.
			
			\item  The semantic meaning of the  symbol may be different under a different knowledge base. For instance, ``notebook"  indicates paper notebook in stationery store  and also means computer in computer store. It depends on the specific context knowledge base.	
		\end{itemize}
		Consequencely, semantic entropy of a random variable depends not only on its probability distribution, but also on the    value  of  the random variable  and its  background knowledge base.

		Note that, the existing semantic entropy definitions are task-oriented, that is, these semantic entropies can only be used in the given tasks.
		To propose a  more general measure of semantic information, we  define the semantic entropy as \emph{the uncertainty of the semantic interpretation of the random variable   in the knowledge base  background.}

		\subsection{Discrete semantic entropy}


		Specifically,  let $U$ denote a   discrete random variable  with the probability mass function ${p}\left( u\right)$.
		Based on knowledge base $\mathcal{K}$, let ${{S}}$ denote the semantic interpretation  of the discrete random variable   $U$.
		The semantic interpretation ${{S}}$    depends on the  value of $U$ and   knowledge base  $\mathcal{K}$.
		Furthermore, the probability   mass function of the semantic interpretation ${{S}}$ is given as, 	{
		\begin{align}
			p\left( {{s}} \right) = \sum_u \limits {p_\mathcal{K}\left( {{s}|u} \right)p\left( u \right)} \label{eq2},
		\end{align}
		where  the conditional probability $p_\mathcal{K}\left( {{s}|u} \right)$ is the probability transition function from  random variable  $U$  to its semantic interpretation ${{S}}$,   which  depends on  the  value of $U$ and   knowledge base  $\mathcal{K}$.}
		Sometimes, it is difficult to measure the relationship $p\left( {{s}|u} \right)$ between  random variable $U$ and its semantic ${{S}}$. Estimating $p\left( {{s}|u} \right)$ is a challenging problem, and one could use statistical methods or deep learning (DL) to estimate.

		\begin{definition}
			\label{Hs}
			The   discrete semantic entropy  ${H_{\text{s}}}\left( U \right)$ of the random variable   $U$ given knowledge base $\mathcal{K}$ is	defined by,
			\begin{align}
				{H_{\text{s}}}\left( U \right) &\triangleq 
				 - \sum\limits_{{s}} { \big(\sum_u \limits {p_\mathcal{K}\left( {{s}|u} \right)p\left( u \right)}\big){{\log }_2}\big(\sum_u \limits {p_\mathcal{K}\left( {{s}|u} \right)p\left( u \right)}\big)}.
			\end{align}
		\end{definition}

		The  relationships between our proposed   semantic entropy and the existing semantic entropy  definitions and Shannon entropy are shown in Table \ref{Relationship}.

\begin{table*}[htbp]
	\caption{The  relationships between our proposed   semantic entropy, the existing semantic entropy definitions and Shannon entropy.}
	\label{Relationship}
	\centering
	\small
	\begin{tabular}{|c|c|c|c|}
		\hline
		\rule{0pt}{1pt}\small{Forms} & \small{Semantic entropy equation}  & \makecell[c]{Relationship with the proposed \\semantic entropy} \\ \hline
		\rule{0pt}{15pt} \small{Proposed semantic entropy} & $ {H_s}\left( X \right) =  - \sum\limits_{{s}} {p\left( {{s}} \right){{\log }_2}p\left( {{s}} \right)} $&   $/$ \\ \hline
		
		\rule{0pt}{15pt} \small{Logical semantic entropy\cite{Basu_PMC_2014}} &$H\left( M \right) =  - \sum\limits_{m \in {\Delta _M}} {{P_M}\left( m \right){{\log }_2}\left( {{P_M}\left( m \right)} \right)} $&  $p\left( {{s}} \right) = {p_M}\left( m \right)$ is logical probability \\ \hline
		
		\rule{0pt}{15pt} \small{Fuzzy semantic entropy\cite{Liu_ITFS_2020}} &$H\left( \zeta  \right) =  - \sum\limits_{j = 1}^m {{D_j}\left( \zeta  \right){{\log }_2}{D_j}\left( \zeta  \right)} $ &  $p\left( {{s}} \right) = {D_j}\left( \zeta  \right)$  is the degree of match \\ \hline
		
		\rule{0pt}{15pt} \small{\makecell[c]{Language understanding\\ semantic entropy \cite{Venhuizen_MDPI_2019}}}  &$ H\left( {{v_t}} \right) =  - \sum\limits_{{{\mathbf{v}}_M} \in {V_\mathcal{M}}}^{} {P\left( {{{\mathbf{v}}_M}|{{\mathbf{v}}_t}} \right)\log P\left( {{{\mathbf{v}}_M}|{{\mathbf{v}}_t}} \right)} $  & \makecell[c] {$p\left( {{s}} \right) = P\left( {{{\mathbf{v}}_M}|{{\mathbf{v}}_t}} \right)$ is the conditional\\ probability of language understanding} \\ \hline
		
		\rule{0pt}{15pt} \small{Shannon entropy  \cite{Cover_book_1999}}  &$ H\left( X \right) =  - \sum\limits_{x \in \mathcal{X}} {p\left( x \right)\log p\left( x \right)} $  & ${p\left( x \right)}$ is the probability
		mass function \\ \hline
		
	\end{tabular}
\end{table*}

		The following two examples illustrate the relationship between semantic entropy and its value and knowledge base.

		\textbf{Example 1.} The random variable $U \in \mathcal{U}$ represents the student mentioned by the teacher, its value space $\mathcal{U}$ has three different values, $\mathcal{U} = \left\{ {{u_1}:{\text{Alice}},{u_2}:{\text{Bob}},{u_3}:{\text{Cindy}}} \right\}$. Assume the probability of ${u_i}$ is shown in Table \ref{tab2}. Then, the  Shannon entropy of random variable $U$ is $H\left( U \right)=1.5{\text{bits}}$.
		\begin{table}[H]
			\caption{ the probability of ${u_i}$.}
			\label{probability}
			\centering
			\begin{tabular}{|c|c|c|c|}
				\hline
				\rule{0pt}{15pt}$U$& ${u_1}$  & ${u_2}$& ${u_3}$ \\ \hline
				\rule{0pt}{15pt} $p\left( {{u_i}} \right)$ & $ 0.25 $& $0.5$ & $0.25$\\ \hline	
			\end{tabular}
			\label{tab2}
		\end{table}

		\textbf{Example 2.} Given the  knowledge base  $\mathcal{K}$, let ${S}$ denote the semantic interpretation of  the random variable $U$  in Example 1, its semantic interpretation space $\mathcal{S} = \left\{ {{s_n}} \right\}_{n = 1}^N$, ${S} \in {\mathcal{S}}$. The transition probability from $U$ to   ${S}$ is written as $p\left( {{s_n}\left| {{u_i}} \right.} \right) = {\Pr _{{\mathcal{K}}}}\left( {{S} = {s_n}|U = {u_i}} \right)$.
		
			\begin{table}[H]
			\caption{ The transition probability from $U$ to   ${S}$.}
			\centering
			\begin{tabular}{|c|@{}c<{\centering}|ccc|}
				\hline
				Condition& \diagbox[width=3em,trim=l]{${S}$}{${U}$} & ${u_1}$  & ${u_2}$& ${u_3}$ \\ \hline
				\rule{0pt}{15pt}
				\multirow{3}*{$\mathcal{K}_1$} &  ${s_1}$ & $\frac{1}{3}$ & $\frac{1}{3}$ & $\frac{1}{3}$ \\
				\rule{0pt}{15pt}
				~ &  ${s_2}$ & $\frac{1}{3}$ & $\frac{1}{3}$ & $\frac{1}{3}$ \\
				\rule{0pt}{15pt}
				~ &  ${s_3}$ & $\frac{1}{3}$ & $\frac{1}{3}$ & $\frac{1}{3}$ \\ \hline
				\rule{0pt}{15pt}
				\multirow{2}*{$\mathcal{K}_2$} &  ${s_1}$ & $\frac{9}{10}$ & $\frac{4}{5}$ & $\frac{1}{2}$ \\
				\rule{0pt}{15pt}
				~ &  ${s_2}$ & $\frac{1}{10}$ & $\frac{1}{5}$ & $\frac{1}{2}$ \\ \hline
				
			\end{tabular}
			\label{tab3}
		\end{table}
	
		\textbf{(1)}  Based on  knowledge base ${\mathcal{K}_1}$: physical education class. Semantic interpretation ${S}$ indicates the instruction that the mentioned student needs to complete, and semantic interpretation space ${\mathcal{S}_1} = \left\{ {{s_1}{\text{: Turn left, }}{s_2}{\text{: Go straight, }}{s_3}{\text{: Turn right}}} \right\}$ and ${S} \in {\mathcal{S}_1}$.
		The corresponding transition probabilities ${p_1}\left( {{s_n}|{u_i}} \right) $ are listed in Table \ref{tab3}.
		According to \eqref{eq2}, the  probabilities of the corresponding semantic interpretation are calculated as  ${p_1}\left( {{s_n}} \right) = \Pr \left( {{S} = {s_n}} \right) = \frac{1}{3},n = 1,2,3$.
		Thus, with the  knowledge base ${\mathcal{K}_1}$, the semantic entropy of  $U$ is ${H_{{{\text{s}}_{\text{1}}}}}\left( U \right) = - \sum\limits_{{s_n}} {{p_1}\left( {{s_n}} \right){{\log }_2}{p_1}\left( {{s_n}} \right)}  = {\log _2}3\left( {{\text{bits}}} \right) $,
		which is larger than its Shannon entropy   $H\left( U \right)=1.5 ({\text{bits}})$.

		\textbf{(2)} Based on  knowledge base ${\mathcal{K}_2}$: the classroom. Semantic interpretation ${S}$ indicates the willingness of the student being questioned. The semantic  space ${\mathcal{S}_2} = \left\{ {{s_1}:{\text{Agree}},{s_2}:{\text{Disagree}}} \right\}$,
		${S} \in {\mathcal{S}_2}$. The corresponding transition probabilities ${p_2}\left( {{s_n}|{u_i}} \right) $ are listed in Table \ref{tab3}. Hence,  ${p_2}\left( {{s_1}} \right) = \frac{3}{4}$ and ${p_2}\left( {{s_2}} \right) = \frac{1}{4}$.	 Further, with the  knowledge base ${\mathcal{K}_2}$, the semantic entropy of $U$ is $	 {H_{{{\text{s}}_{\text{2}}}}}\left( U \right) =  - \sum\limits_{{s_n}} {{p_2}\left( {{s_n}} \right){{\log }_2}{p_2}\left( {{s_n}} \right)} \nonumber= 2 - \frac{3}{4}{\text{lo}}{{\text{g}}_{\text{2}}}{\text{3}}\left( {{\text{bits}}} \right)$,
		which is lower than its Shannon entropy   $H\left( U \right)=1.5 ({\text{bits}})$.

		Since  the semantic information of   $U$   relies on the distribution $p\left( {{s}} \right)$,
		the  proposed discrete semantic entropy ${H_{\text{s}}}\left( U \right)$ of $U$   may be higher, lower than or equal to  Shannon entropy $H\left( U \right)$.

		\subsection{Differential     semantic entropy}
		
		Let $S$ denote the semantic random variable of the  symbol  $U$ with probability density function $f\left( s \right)$, and the corresponding mapping function is,
		The semantic interpretation   depends on the  value of $U$ and   knowledge base  $\mathcal{K}$.
		Furthermore, the probability density function of the semantic interpretation $S$ is given as,
		\begin{align}
			f\left( s \right) = \int_u {f\left( {s|u} \right)f\left( u \right)du},
		\end{align}
		where     $f\left( {s|u} \right)$ is the probability transition function from  random variable  $U$  to its semantic interpretation $S$,   which  depends on  the  value of $U$ and   knowledge base  $\mathcal{K}$.
		Thus, the differential semantic entropy  ${h_{\text{s}}}\left( U \right) $ of the  symbol   $U$ is	 defined as,
		\begin{align}
			{h_{\text{s}}}\left( U \right) \triangleq h\left( S \right) =  - \int\limits_s {f\left( s \right){{\log }_2}f\left( s \right)ds}.
		\end{align}

		\subsection{Measure of semantic knowledge  }
		
		Let $K$ represent the semantic knowledge variable with probability distribution    $p\left( {k} \right)$.
		Thus,  the entropy  of semantic knowledge $K$ is given as,
		\begin{align}
			H\left( K \right) =  - \sum\limits_{k} {p\left( {k} \right){{\log }_2}p\left( {k} \right)}.
		\end{align}
		

		It is well known that semantic coding can provide extremely high compression ratios, but the specific measure of semantic compression is still unknown. Next, we will try to explain the reasons for the high compression ratio of semantic encoding.
		The relationships among  the random variables $U$, its semantic interpretation   $S$, and semantic knowledge  $K$   are given as,
		\begin{align}
			H\left( U \right) &=
			  {I\left( K; U\right)} +   {I\left( S;U|K \right)} +   {H\left( {U\left| {K,S} \right.} \right)} . \nonumber\\
			& =  {I\left( K; U\right)} +   {H( S|K)}  -H\left(S|U,K \right)+   {H\left( {U\left| {K,S} \right.} \right)} .\nonumber\\
		\end{align}
		where $I\left( K; U\right)$ represents the shared  related knowledge between the transmitter and destination; $H(S|K) $ represents the conditional semantic entropy given the knowledge variable $K$, equaling to the semantic entropy $H(S)$ if $S$ and $K$ are independent; $H\left( {U\left| {K,S} \right.} \right)$ represents the remaining information after extracting $K,S$.
		
		Furthermore, let ${G_{\text{s}}}$  denote semantic compression gain, i.e., the ratio between Shannon entropy and semantic entropy of the  data, w.t,
		\begin{align}
			{G_{\text{s}}} = \frac{{H\left( U \right)}}{{H_{\text{s}}}\left( U \right)}.
			\end{align}
		Thus, the amount of    transmission    will be decreased with the increase of semantic compression gain ${G_{\text{s}}}$.
		Since the transmitter and receiver share semantic knowledge, and semantic coding removes semantic redundancy, semantic communications has a  high compression ratio. Our goal is to find the maximum achievable compression gain without semantic error.
		
		\textbf{Example 3. (Inference task-oriented semantic communication)}
		Consider a typical task-oriented semantic communication system, i.e., the classification task for the MNIST dataset with a set of handwritten character numbers from 0 to 9, which are the semantic information of the MNIST classification task.
		The semantic entropy of each image is equal to the Shannon entropy of   labels, i.e.,
		${\log _2}10 {\text{ bits}}$.
		The average size of each image of the MNIST dataset is about $6272{\text{ bits}}$.
		Assume that semantic encoding scheme can perfectly exact the semantic information, the maximum semantic compression gain ${G_{\text{s}}}$ is given as,
		\begin{align}
			{G_{\text{s}}} = \frac{{6272}}{{{{\log }_2}10  }} \approx 1888.
		\end{align}

		\textbf{Example 4. (Data reconstruction semantic communication)}
		Consider a typical data reconstruction semantic communication system with fourier transform (FT) and inverse fourier transform (IFT). Specifically, the semantic encoder and
		decoder are designed by FT and IFT, respectively.
		
		Assume that the input signal  is a single-frequency sine wave, i.e., $f\left( t \right) = \sin {\Omega_0}t$, and  the amount of information are
		$N {\text{ bits}}$ (Since the number of discrete points in the time domain
		may tend to be infinite, the value of $N$ may tend to be infinite).
		
		Via the FT-based semantic encoder,  the frequency domain expression of the input signal is given as,
		\begin{align}
			F\left( \Omega  \right) = j\pi \left[ {\delta \left( {\Omega  + {\Omega _0}} \right) - \delta \left( {\Omega  - {\Omega _0}} \right)} \right].
		\end{align}
		Note that ${\Omega _0} \in {\omega _0}$ is the semantic information, the  semantic entropy is  $\operatorname{H} \left( {{\Omega _0}} \right) = \sum\limits_{{\Omega _0} \in {\omega _0}} { - p\left( {{\Omega _0}} \right)} $ ${\log _2}p\left( {{\Omega _0}} \right)$  bits.
		Generally, $\operatorname{H} \left( {{\Omega _0}} \right)$ is a limited value.
		The parameters $j$, $\pi $ and $\delta \left( \Omega \right)$ is the semantic knowledge base shared by  the transmitter and receiver.
		Therefore, the semantic compression gain  is ${G_{\text{s}}} = \frac{N}{{\operatorname{H} \left( {{\Omega _0}} \right)}}$, which increases as the number of bits $N$ increases.

		\section{Semantic channel coding theorem}
		
		\subsection{Semantic communication system model}
	\begin{figure}[ht]
		\centering
		\includegraphics[height=5cm]{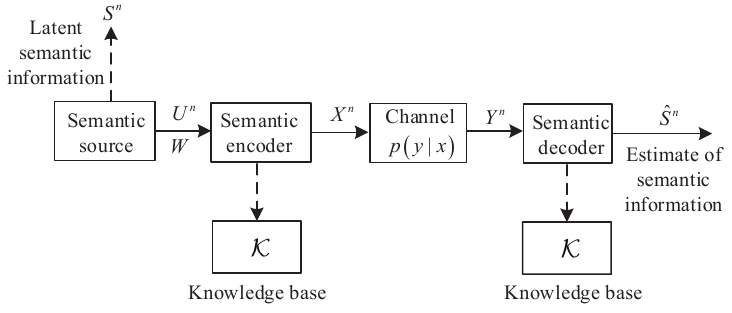}
	
		\caption{~A semantic  communication system}
		\label{img0}
	\end{figure}
		
		Consider a  semantic  communication system, as  depicted in Fig. \ref{img0}, where $\mathcal{K}$ represents the semantic knowledge base shared by the sender and receiver. 
		Assume that the source sequence ${U^n}$  and semantic information sequence ${S^n}$ satisfies the  ``many-to-one" mapping relation.   For example,  in the MNIST dataset, ${U^n}$ denotes the input images, and ${S^n}$ represents the  label of the input images, which is from 0 to 9, and there are multiple images with the same label. Specifically,  as shown in Fig. \ref{img1},    $\mathcal{U}$ and $\mathcal{S}$  denote  the source sequence set and   the semantic sequence set, respectively, and $u_l^n$ and $s_t^n$ denote the $l$-th source sequence in $\mathcal{U}$ and $t$-th semantic sequence in  $\mathcal{S}$, respectively, where $l \in \left\{ {1, \ldots ,L} \right\}$ and $t \in \left\{ {1, \ldots ,T} \right\}$. Moreover, ${\mathcal{U}_{{s_t}}}$ is the set of source sequences has the same semantic meaning $s_t^n$.
		Furthermore,  the corresponding transition probability $\sum\limits_s {p\left( {{s^n}\left| {{u^n}} \right.} \right) = 1} $ for all $u^n\in\mathcal{U}_s$, where $\mathcal{U}_s$ denotes the set of source sequences has the same semantic meaning $s^n$.
		Define a message $W \in \left[ {1:{2^{n\operatorname{H} \left( U \right)}}} \right]$ with a ``one-to-one" correspondence to the source $U$.
		The semantic transmitter wishes to reliably transmit the latent semantic information $S$ within the message $W$  at a rate $R$ bits per transmission to
		the semantic receiver over a noisy physical channel.
		
		\begin{figure}[ht]
			\centering
			\includegraphics[height=6cm]{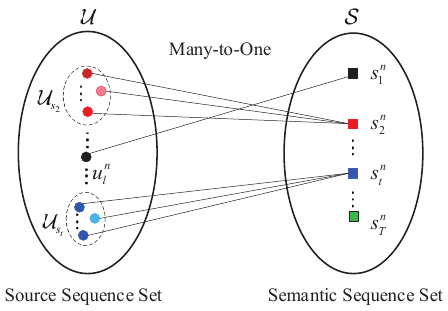}
			\caption{~The ``many-to-one" relations between ${U^n}$ and ${S^n}$.}
			\label{img1}
		\end{figure}

		Toward this
		end, the semantic transmitter encodes the message $W$ into a codeword $X^n$ and transmits it over the channel in $n$ time   channel uses.
		Upon receiving the noisy sequence $Y^n$, the semantic receiver decodes it to obtain the estimate $\widehat S$ of the semantic message.
		The semantic channel coding problem is concerned with finding     the highest rate $R$ such
		that the probability of semantic decoding error can be made to decay asymptotically to zero with
		the code block length $n$.

		Consider the sematic channel coding problem for a typical discrete memoryless channel
		(DMC) model $p\left( {y|x} \right)$  with a finite input set  ${\mathcal{X}}$ and  a finite output set ${\mathcal{Y}}$.
		
		The source message set $\left[ {{\text{1}}:{2^{nR}}} \right] = \left\{ {{\text{1}},{\text{2}},...,{2^{nR}}} \right\}$ consists a semantic message set $\left[ {{\text{1}}:{2^{\left\lceil {\alpha nR} \right\rceil }}} \right] \\= \left\{ {{\text{1}},{\text{2}},...,{2^{\left\lceil {\alpha nR} \right\rceil }}} \right\}$,   where the coefficient $\alpha $ determines the ratio of many to one, and lies in the range $0 \leq \alpha  < 1$.
	
		A semantic encoder  assigns a codeword ${X^n}\left( m \right)$ to
		message $W \in \left[ {{\rm{1 : }}{2^{ nR}}} \right]$,   where $m \in \left[ {1:{2^{\alpha nR}}} \right]$ is the index of the message set with the same semantic meaning.
		Furthermore, let
		${\Phi _m}$ denote the subset of message $W$ with the same semantic information, $ m \in \left[ {1:{2^{\alpha nR}}} \right]$, where ${\Phi _m}$ satisfies ${\Phi _m} \subseteq \left[ {{\rm{1 : }}{2^{nR}}} \right]$,   ${\Phi _m} \cap {\Phi _l} = \emptyset ,\forall m \ne l$ and $\bigcup\nolimits_{m = 1}^{{2^{\alpha nR}}} {{\Phi _m} = \left[ {{\rm{1 : }}{2^{nR}}} \right]} $.
		Then, a semantic decoder   assigns an estimate $\widehat m \in \left[ {1:{2^{\alpha nR}}} \right]$ or an error message $\varepsilon$ to each received sequence ${Y^n}$, $ {\mathcal{Y}^n}  \in  \left[ {{\rm{1 : }}{2^{\alpha nR}}} \right] \cup \left\{ \varepsilon  \right\}$.
		We assume that the semantic message is uniformly distributed over the message set $\left[ {1:{2^{\alpha nR}}} \right]$.
		

		The performance of a given code is measured by the probability that the estimate of
		the semantic message is different from the actual semantic message sent. The set $\mathcal{C} = \left\{ {{x^n}\left( 1 \right),{x^n}\left( 2 \right), \ldots ,} \right.$ $\left. {{x^n}\left( {{2^{\left\lceil {\alpha nR} \right\rceil }}} \right)} \right\}$is referred to as the codebook associated with the $\left( {{2^{\alpha nR}},n} \right)$ code. More precisely, let ${\lambda _m}\left( {\cal C} \right) = {\rm{Pr}}\left\{ {\widehat S \ne m|S = m} \right\}$ be the conditional probability of error given  semantic  index $m$ is sent. Then, the
		average probability of error for a $\left( {{2^{\alpha nR}},n} \right)$ code is,
		\begin{align}P_{{\text{e,s}}}^{\left( n \right)}\left( \mathcal{C} \right) = {\text{Pr}}\left\{ {\widehat S \ne S} \right\} = \frac{1}{{{2^{\left\lceil {\alpha nR} \right\rceil }}}}\sum\limits_{m = 1}^{{2^{\left\lceil {\alpha nR} \right\rceil }}} {{{\text{E}}_\mathcal{C}}\left( {{\lambda _m}\left( \mathcal{C} \right)} \right)}. 	
		\end{align}


		\subsection{Semantic channel coding theorem}


		Given the semantic mapping   ${f_s}:\left[ {1:{2^{nR}}} \right] \to \left[ {1:{2^{\alpha nR}}} \right]$ and
		the discrete memoryless  channel $p\left( {y\left| x \right.} \right)$, we define the   semantic channel capacity  ${C_{\text{s}}}$ as
		\begin{align}
			{C_{\text{s}}} = \mathop {\max }\limits_{p\left( x \right)} \frac{{I\left( {X;Y} \right)}}{\alpha }.
		\end{align}
		\begin{theorem}{\textbf{Semantic channel coding theorem:}}
			For every bit rate $R < {C_{\text{s}}} = {\max _{p\left( x \right)}}\frac{{I\left( {X;Y} \right)}}{\alpha }$, there exists a sequence of $\left( {{2^{nR}},n} \right)$ codes with average probability of error $P_{{\text{e,s}}}^{\left( n \right)}$ that tends to zero as $n \to \infty $.\\
			Conversely, for every sequence of  $\left( {{2^{nR}},n} \right)$ codes with probability of error $P_{{\text{e,s}}}^{\left( n \right)}$ that tends to zero as $n \to \infty $, the rate must satisfy $R < {C_{\text{s}}} = {\max _{p\left( x \right)}}\frac{{I\left( {X;Y} \right)}}{\alpha }$.
		\end{theorem}
		
		Note that, different form the  proof of    Shannon  channel coding theorem,
		the semantic information mapping should be incorporated in both semantic encoding and decoding in the
		achievability proof of proposed semantic coding theorem.
		We first prove the achievability in the following subsection. The proof of the converse is
		given in Section IV. B
		
		\begin{proof}
			For simplicity of presentation, we assume throughout the proof that $nR$, $\alpha nR$ and $n{R_{\text{s}}}$  are  integers.
			
			\textbf{Semantic source:} Given a source message  $W \in \left[ {{\rm{1 : }}{2^{nR}}} \right] $.  Let $S$   represent the inherent semantic information of $W$.
			Let ${f_{\text{s}}}:\left[ {1:{2^{nR}}} \right] \to \left[ {1:{2^{\alpha nR}}} \right]$ represent the semantic mapping from $W$ to $S$.

			\textbf{Random codebook generation:} By applying random coding, we randomly and independently generate ${2^{n{R_{\text{s}}}}}$ sequences ${x^n}\left( m \right)$, each according to $p\left( {{x^n}} \right) = \prod\limits_{i = 1}^n {{p_X}\left( {{x_i}} \right)}$. The generated sequences constitute the codebook  $\mathcal{C} = \left[ {{\rm{1 : }}{2^{n{R_{\text{s}}}}}} \right]$ as follows
			\begin{align}
				p\left( \mathcal{C} \right) = \prod\limits_{m = 1}^{{2^{n{R_{\text{s}}}}}} {\prod\limits_{i = 1}^n {{p_X}\left( {{x_i}\left( m \right)} \right)} } .
			\end{align}
			The   codebook $\mathcal{C}$ is known by  both the semantic encoder and the  semantic decoder.
			In order to  represent ${2^{\alpha nR}}$ semantic information losslessly, ${R_{\text{s}}}$ must satisfy
			\begin{align}
				{R_{\text{s}}} \geq   \alpha R.
			\end{align}
			
			\textbf{Semantic encoding:} Given a source message   ${W}$,  the encoder finds the corresponding semantic information set index $m$, i.e., to send semantic  index $ m \in \left[ {1:{2^{n{R_{\text{s}}}}}} \right]$, and transmits ${x^n}\left( m \right)$.
			
			\textbf{Semantic  decoding:}  Let ${y^n}$ be the received sequence.
			The receiver declares that $\widehat m $ is sent if  ${x^n}\left( {\widehat m} \right)$ and ${y^n}$ are jointly typical; otherwise, if there is none or more than one such message, it declares an error  $\mathcal{E}$.

			\textbf{Analysis of the probability of error:}	Assume that the sending   semantic index $m$ obeys a uniform distribution  on $m \in \left[ {1:{2^{n{R_{\text{s}}}}}} \right]$. Without loss of generality, we assume  that semantic index $m = 1$ is sent and typical sequence   ${x^n}\left( 1 \right)$ is sent, received as ${y^n}$, and thus the decoder makes an error if $\left( {{x^n}\left( 1 \right),{y^n}} \right)$ are not jointly typical.
			Let $\mathcal{E} = \left\{ {\widehat m \ne m} \right\}$ denote the semantic error event. Consider the probability of semantic errors averaged over all codewords and codebooks, we have
			\begin{subequations}
				\begin{align}
					{\text{Pr}} \left( \mathcal{E} \right) &= {{\rm E}_\mathcal{C}}\left( {P_{{\text{e,s}}}^{\left( n \right)}} \right)\label{SE_a} \\
					&= {{\rm E}_\mathcal{C}}\left( {\frac{1}{{{2^{n{R_{\text{s}}}}}}}\sum\limits_{m = 1}^{{2^{n{R_{\text{s}}}}}} {{\lambda _m}\left( \mathcal{C} \right)} } \right) \label{SE_b} \\
					&= \frac{1}{{{2^{n{R_{\text{s}}}}}}}\sum\limits_{m = 1}^{{2^{n{R_{\text{s}}}}}} {{{\rm E}_\mathcal{C}}\left( {{\lambda _m}\left( \mathcal{C} \right)} \right)}  \label{SE_c} \\
					&=   {{\rm E}_\mathcal{C}}\left( {{\lambda _1}\left( \mathcal{C} \right)} \right) \label{SE_d} \\
					&= \Pr \left( {\mathcal{E}\left| {m = 1} \right.} \right),
				\end{align}
			\end{subequations}
			where $\lambda _m $ is the conditional probability of semantic error given semantic index $m$ was sent.
			Equation \eqref{SE_d} holds due to the symmetry of the random codebook generation.
			

			The decoder makes an error if one or both of the following events occur:
			\begin{itemize}	
				\item  $E_1^c$ is defined as   $\left( {{X^n}\left( 1 \right),{Y^n}} \right)$  not   jointly typical;
				\item  ${E_m}$ is defined as  $\left( {{X^n}\left( m \right),{Y^n}} \right)$   semantic jointly typical, where $m \ne 1$.	
			\end{itemize}
			Thus, by the union of events bound
			\begin{subequations}
				\begin{align}
					{\text{Pr}}\left( \mathcal{E} \right) &= P\left( {E_1^c \cup {{\left\{ {{E_m}} \right\}}_{m \ne 1}}} \right) \label{EB_a}\\
					&\leq  P\left( {E_1^c} \right) + P\left( {{{\left\{ {{E_m}} \right\}}_{m \ne 1}}} \right) \label{EB_b}\\
					&\leq \varepsilon  + \sum\limits_{m \ne 1} {P\left( {{E_m}} \right)}  \label{EB_c}\\
					&\leq \varepsilon  + \sum\limits_{m \ne 1} {{2^{ - n\left( {I\left( {X;Y} \right) - 3\varepsilon } \right)}}}  \label{EB_d}\\
					&= \varepsilon  + \left( {{2^{n{R_{\text{s}}}}} - 1} \right){2^{ - n\left( {I\left( {X;Y} \right) - 3\varepsilon } \right)}} \label{EB_f}\\
					&\leq \varepsilon  + {2^{3n\varepsilon }}{2^{ - n\left( {I\left( {X;Y} \right) - {R_{\text{s}}}} \right)}}  \label{EB_g}\\
					&\leq 2\varepsilon\label{EB_g},
				\end{align}
			\end{subequations}
			where inequality \eqref{EB_b} holds due to  the probabilities of union  events bound, inequality \eqref{EB_c} holds due to    $P\left( {E_1^c} \right) \to 0$  as $n \to \infty$,
			 inequality (15d) holds  due to the  property of the joint AEP, i.e., the probability that ${X^n}\left( m \right)$ and ${Y^n}$  are jointly typical is less than  $ {2^{ - n\left( {I\left( {X;Y} \right) - 3\varepsilon } \right)}}$ (that is, ${P\left( {{E_m}} \right)}\leq  {2^{ - n\left( {I\left( {X;Y} \right) - 3\varepsilon } \right)}}$)\cite{Cover_book_1999},
			 inequality \eqref{EB_f} holds  due to $\left( {{2^{n{R_{\text{s}}}}} - 1} \right) \to {2^{n{R_{\text{s}}}}}$ as $n \to \infty$, and
			inequality \eqref{EB_g} holds  because   ${R_{\text{s}}} \leq I\left( {X;Y} \right)$.

			Furthermore, by combining  ${R_{\text{s}}} \geq   \alpha R$   and ${R_{\text{s}}} \leq I\left( {X;Y} \right)$, we can obtain
			\begin{align}R \leq \frac{{I\left( {X;Y} \right)}}{\alpha }.
			\end{align}
			Hence, there   exists a sequence of $\left( {{2^{nR}},n} \right)$  codes such that ${\lim _{n \to \infty }}P_{{\text{e,s}}}^{\left( n \right)} = 0$, which proves that $R < {C_{\text{s}}}= {\max _{p\left( x \right)}}\frac{{I\left( {X;Y} \right)}}{\alpha }$ is achievable. Finally, taking $\varepsilon  \to 0$ completes the proof.
		\end{proof}	
		
		\section{Converse Proof}
		Fano's inequality is the key to proving the converse  of the channel coding theorem \cite{Cover_book_1999}.
		However, since different symbols may have the same semantic information, the traditional bit error judgment rule cannot be directly used for semantic symbol error judgment, which makes Fano's inequality not applicable to semantic communications.

		\subsection{ Extended Fano's inequality}
		To prove the converse  of the semantic channel coding theorem, we first derive and prove the semantic Fano's inequality.
		Specifically,
		the message $W$ is uniformly distributed on the set    $\left[ {{\rm{1 : }}{2^{ nR}}} \right] \triangleq \Psi$, $W \in \Psi$ and the sequence ${Y^n}$ is related probabilistically to $W$.
		From ${Y^n}$, we estimate that the message $W$  was sent. Let the estimate be $\widehat W = g\left( {{Y^n}} \right)$.
		Thus, $ W \to {X^n}\left( W \right) \to {Y^n} \to \widehat W $ forms a Markov chain.
		Define the semantic error probability $P_{{\text{e,s}}}^{\left( n \right)} \triangleq \Pr \left\{ {\widehat W \notin {\Phi _W}} \right\}$.
		For semantic communication systems, assume that the number of messages with the same semantic information as $W$ is ${2^{\left( {1 - \alpha } \right)nR}} = \beta \left| \Psi  \right|$, and the set of messages  is denoted as ${\Phi _W}$, $\left| {{\Phi _W}} \right| = {2^{\left( {1 - \alpha } \right)nR}}$, where $\alpha $  and $\beta $ are coefficients which satisfy $0 \leq \alpha  < 1$, $0 \leq \beta  < 1$, respectively.
		
		\begin{theorem}
			\textbf{Semantic Fano's inequality:}
			Given message $W \in \left[ {{\rm{1 : }}{2^{ nR}}} \right]$	with semantic $\left[ {1:{2^{nR}}} \right] \to \left[ {1:{2^{\alpha nR}}} \right]$. Let $P_{{\text{e,s}}}^{\left( n \right)}= \Pr \left\{ {\widehat {W} \notin {\Phi _{W}}} \right\}$, we have
			\begin{align}\label{Fano}
				H\left( {W|{Y^n}} \right) \leq 1 + \left( {1 - \alpha  + \left( {\gamma  + \alpha  - 1} \right)P_{{\text{e,s}}}^{\left( n \right)}} \right)nR.
			\end{align}
		\end{theorem}
		\begin{proof}
			Define a semantic error random variable $E$ as follows
			\begin{align}
				E = \left\{ \begin{gathered}
					\begin{array}{*{20}{c}}
						1&{{\text{if }}\widehat W \notin {\Phi _W}}
					\end{array} \hfill \\
					\begin{array}{*{20}{c}}
						0&{{\text{if }}\widehat W \in {\Phi _W}}
					\end{array} \hfill \\
				\end{gathered}  \right..
			\end{align}
			Then, by applying the chain rule for entropies to expand $H\left( {E,W|{Y^n}} \right)$ in two different ways, we have
			\begin{subequations}
				\begin{align}
					H\left( {E,W|{Y^n}} \right) &= H\left( {W|{Y^n}} \right) + H\left( {E|W,{Y^n}} \right) \\
					&= H\left( {E|{Y^n}} \right) + H\left( {W|E,{Y^n}} \right).
				\end{align}
			\end{subequations}
			Because conditioning reduces entropy, $H\left( {E|{Y^n}} \right) \leq H\left( E \right)$. Moreover,
			due to $E$ being a function of $W$ and ${Y^n}$, the conditional entropy $H\left( {E|W,{Y^n}} \right) = 0$. Furthermore, since $E$ is a binary-valued random variable, $H\left( E \right) \leq 1$. Thus, $H\left( {W|E,{Y^n}} \right)$ is bounded as follows
			\begin{subequations}
				\begin{align}
					H\left( {W|E,{Y^n}} \right)
					= &\Pr \left( {E = 0} \right)H\left( {W|{Y^n},E = 0} \right)
					+ \Pr \left( {E = 1} \right)H\left( {W|{Y^n},E = 1} \right) \label{Fano_a}\\
					\leq& \left( {1 - P_{{\text{e,s}}}^{\left( n \right)}} \right){\log _2}\left( {\beta \left| \Psi  \right|} \right)
					+ P_{{\text{e,s}}}^{\left( n \right)}{\log _2}\left( {\left( {1 - \beta } \right)\left| \Psi  \right|} \right) \label{Fano_b}\\
					= &\left( {1 - P_{{\text{e,s}}}^{\left( n \right)}} \right)\left( {1 - \alpha } \right)nR + P_{{\text{e,s}}}^{\left( n \right)}\gamma nR \label{Fano_c}\\
					= &\left( {1 - \alpha  + \left( {\gamma  + \alpha  - 1} \right)P_{{\text{e,s}}}^{\left( n \right)}} \right)nR,\label{Fano_d}
				\end{align}
			\end{subequations}
			where  $\gamma  \triangleq \frac{1}{{nR}}{\log _2}\left( {1 - \beta } \right) + 1$,   inequality \eqref{Fano_b} holds since $E = 0$, $\widehat W \in {\Phi _W}$ and  $E = 1$, $\widehat W \notin {\Phi _W}$. The upper bound on the conditional entropy is the log of the number of possible outcomes of $W$. Combining these results, we obtain the semantic Fano inequality  \eqref{Fano}.
		\end{proof}

		\subsection{Converse to semantic channel coding theorem}
		To prove the converse to the channel coding
		theorem, we need to show that for every sequence of $\left( {{2^{nR}},n} \right)$ codes with ${\lim _{n \to \infty }}P_{{\text{e,s}}}^{\left( n \right)} = 0$, we must have $R < {C_{\text{s}}} = {\max _{p\left( x \right)}}\frac{{I\left( {X;Y} \right)}}{\alpha }$.
		
		\begin{proof}
			For a given semantic encoding ${X^n}\left(  \cdot  \right)$ and a given semantic  decoding rule $\widehat W = g\left( {{Y^n}} \right)$, we have $ W \to {X^n}\left( W \right) \to {Y^n} \to \widehat W $. For each $n$, let $W$ be drawn according to a uniform distribution over $\left[ {{\rm{1 : }}{2^{ nR}}} \right]$. Define $P_{{\text{e,s}}}^{\left( n \right)} = \Pr \left\{ {\widehat W \notin {\Phi _W}} \right\}$, hence, we have
			\begin{subequations}
				\begin{align}
					nR =& H\left( W \right) \label{conv_a}\\
					=& H\left( {W|{Y^n}} \right) + I\left( {W;{Y^n}} \right) \label{conv_b}\\
					\leq& H\left( {W|{Y^n}} \right) + I\left( {{X^n}\left( W \right);{Y^n}} \right) \label{conv_c}\\
					\leq& 1 + \left( {1 - \alpha  + \left( {\gamma  + \alpha  - 1} \right)P_{{\text{e,s}}}^{\left( n \right)}} \right)nR
					+ I\left( {{X^n}\left( W \right);{Y^n}} \right)\label{conv_d}\\
					\leq& 1 + \left( {1 - \alpha  + \left( {\gamma  + \alpha  - 1} \right)P_{{\text{e,s}}}^{\left( n \right)}} \right)nR + n{C},\label{conv_f}
				\end{align}
			\end{subequations}
			where \eqref{conv_a} holds due to the assumption that $W$ is uniform over  $\left[ {{\rm{1 : }}{2^{ nR}}} \right]$, inequality \eqref{conv_c} holds due to the data-processing inequality, inequality \eqref{conv_d} holds due to   semantic Fano's inequality, and  inequality \eqref{conv_d} holds due to  $I\left( {{X^n}\left( W \right);{Y^n}} \right) \leq n{C}$. Dividing by $n$, we obtain
			\begin{align}
				R \leq &\frac{1}{{\left( {\alpha  - \left( {\gamma  + \alpha  - 1} \right)P_{{\text{e,s}}}^{\left( n \right)}} \right)n}}
				+ \frac{I\left( {X;Y} \right)}{{\alpha  - \left( {\gamma  + \alpha  - 1} \right)P_{{\text{e,s}}}^{\left( n \right)}}}.
			\end{align}
			Now letting $n \to \infty $ and $P_{{\text{e,s}}}^{\left( n \right)} \to 0$, and hence
			\begin{align}
				R \leq \frac{I\left( {X;Y} \right)}{\alpha }.
			\end{align}
			As a result,  for every sequence of $\left( {{2^{nR}},n} \right)$ codes with ${\lim _{n \to \infty }}P_{{\text{e,s}}}^{\left( n \right)} = 0$, we must have $R < {C_{\text{s}}} = {\max _{p\left( x \right)}}\frac{{I\left( {X;Y} \right)}}{\alpha }$,	which completes the proof of the converse.
		\end{proof}	
		
		\section{Experiments and Discussions}

		
		\begin{figure}[ht]
			\centering
			\includegraphics[width=11cm]{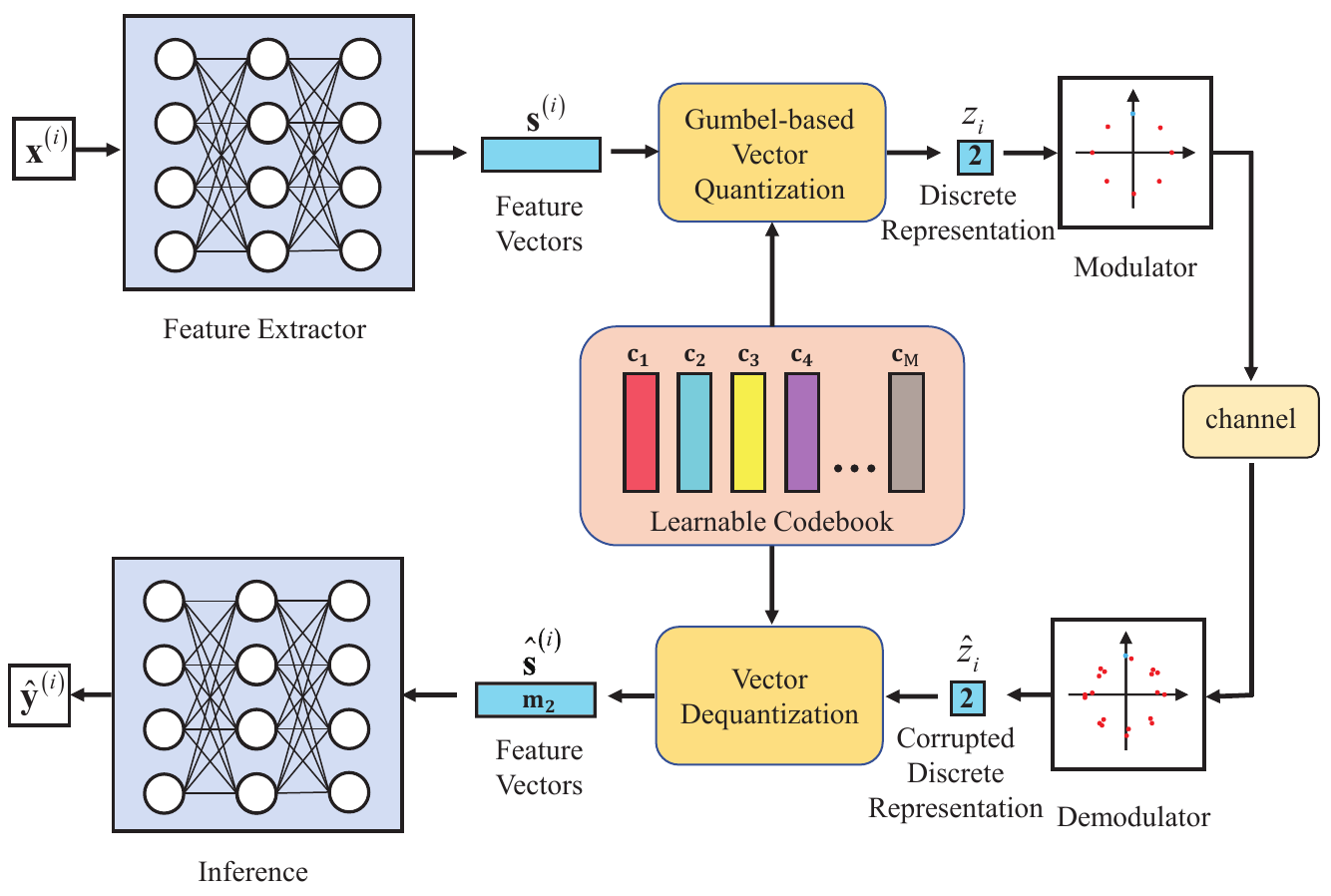}
			\caption{~ A task-oriented semantic communication system.}
			\label{img3}
		\end{figure}
		According to the proposed definition of general semantic entropy, the interpretation of the semantic variable of the classification task is the category. Therefore, the source semantic entropy of the classification task is the Shannon entropy of the category.
		In order to further verify the effectiveness of the definition of semantic entropy and rationality of the proposed semantic channel coding theorem, we adopt a classification task-oriented semantic communication system  \cite{Xie_2023} with the MNIST dataset.
		As is shown in Fig. \ref{img3}, based on the DL network, the feature extractor first extracts   feature vectors from the input image. Furthermore, according to the shared learnable  codebook, the Gumbel-based
		vector
		quantization converts the extracted feature vectors into   discrete representations, and then performs digital modulation and sends it to the receiver. The receiver first demodulates the received signal to corrupted discrete representations, and further performs vector dequantization to obtain feature vectors. Finally, the  feature vectors are decoded by DL-based inference   network, and output the classification information.

		\subsubsection{Semantic entropy verification}
		Assuming that the number of categories of input images     is $N$, and thus the source semantic entropy is ${\log _2}N$. Moreover, we adopt $M$-PSK digital modulation scheme.

		In order to verify the effectiveness of semantic entropy, we first consider the semantic source coding, ignoring the channel noise, that is, the DL-based encoder network and the  DL-based decoder network are jointly trained under the noise-free channel condition.
		\begin{figure}[ht]
			\centering
			\includegraphics[width=9.5cm]{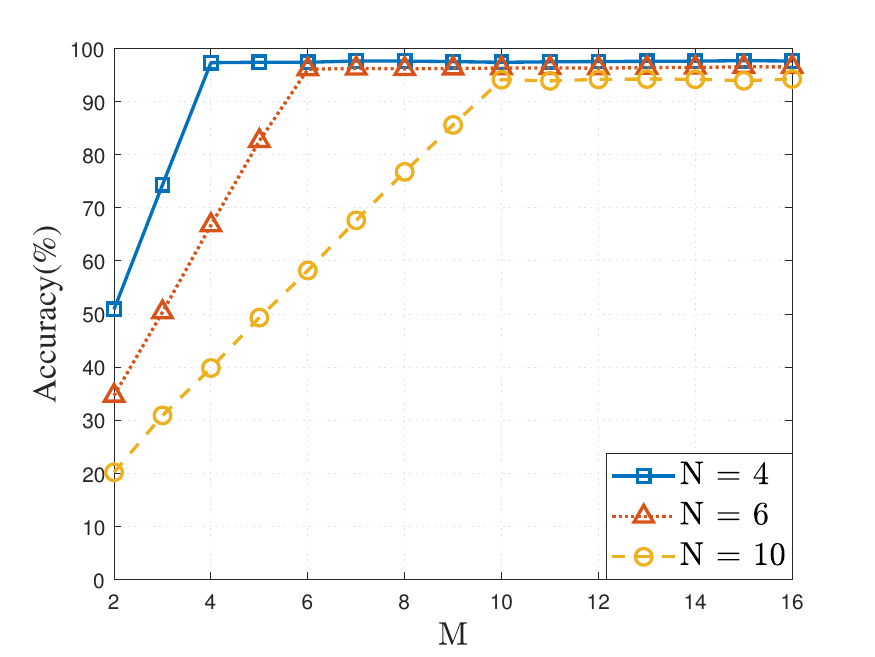}
			\caption{~ Performance of the MNIST $N$-classification task}
			\label{img4}
		\end{figure}

		Fig. \ref{img4} shows  the relationship between receiver inference accuracy performance and modulation order $M$-PSK with    input image categories   $N = 4,6,10$.
		From the  Fig. \ref{img4}, it can be observed that when the modulation order $M$ is smaller than the number of  input image categories $N$,  the inference accuracy   increases as the modulation order $M$ increases, because as $M$ increases, the modulated bits can represent more source semantic information;
		when $M \ge N$, the inference accuracy   reaches $95\%$ and remains unchanged, because
		the semantic information of the information source is fully represented by  the modulated bits. Therefore, Fig.  \ref{img4} verifies that the source semantic entropy are ${\log _2}N$ bits.
	Fig. \ref{img4} also shows the effect of  the knowledge bases on inference performance. Note that $M$  also denotes  the number of learnable codebook (knowledge base).   For the $N=4$ case, when $M=2$, the classification accuracy is only $50\%$, while when $M\ge4$, the classification accuracy is  more than $95\%$. This is because, when $M$ is 2, the knowledge base shared by the transmitter and receiver only contains two types of label information, which is not enough for the 4 labels classification tasks. When $M\ge4$, the knowledge base shared by the transmitter and receiver contains more than four   types of label information, which means that the label information has been
  	fully shared, and thus higher classification accuracy can be
  	achieved.

Moreover,  we     verified the semantic entropy   through the   CIFAR10 dataset with  image categories $N = 4, 6, 10$.
The semantic encoder network includes a three layers convolutional neural network and a two layers ResNet block, and the encoded semantic feature is  quantized by the binary quantization, where $L$ denotes the number of quantization bits.
The semantic decoder includes a two layers of ResNet block, one    maxpool layer and one dense layer.  As shown in Fig.\ref{img}, as   the   number of quantization bits $L$ increases, the inference accuracy  increases first, and then it remains almost unchanged.
The reason is that as      the   number of quantization bits $L$ increases, the source semantic information contained in semantic encoded feature increases, and when $L$ exceeds one threshold, the semantic coding feature  contain all the source semantic information.
For $N=4$, the source semantic entropy is $2$ bits, which is  the limit of semantic compression coding, and  the inference accuracy for the CIFAR10 dataset can reach  $90\%$ with  $L=16$ quantization bits.  For the same number of quantization bits $L$, the inference accuracy     decreases with the increase of  image categories $N$.

	\begin{figure}[ht]
		\centering
		\includegraphics[width=9.5cm]{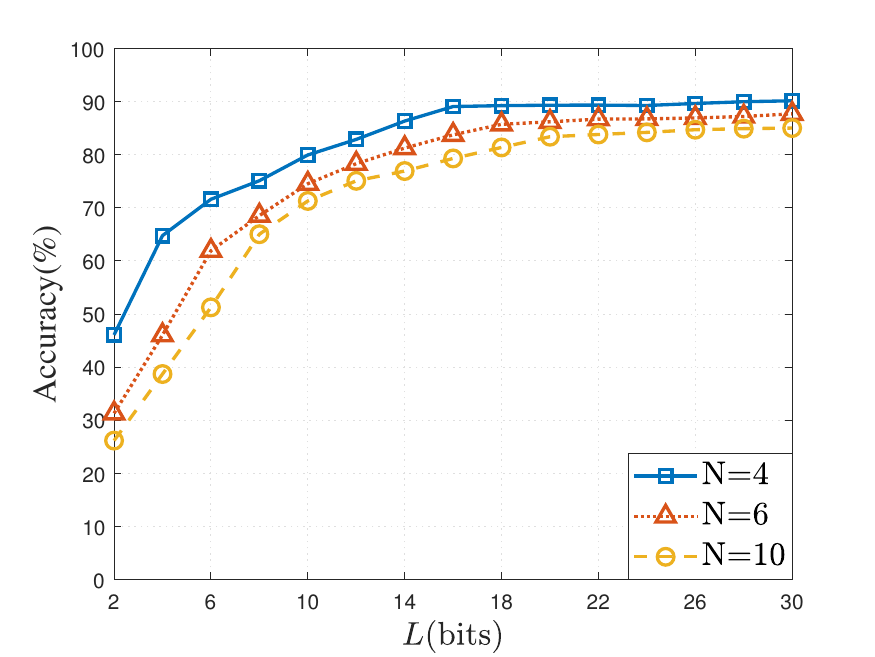}
		\caption{~ Performance of the CIFAR10 $N$-classification task}
		\label{img}
	\end{figure}

		%

		Fig. \ref{img5} (a) and  (b) shows the inference performance of   classification task for the MNIST dataset under different SNRs  with $N=6$ and $10$, respectively.
		From Fig. \ref{img5} (a) and  (b), we observe that the inference accuracy increases first and then remains constant as the SNR increases.
		When the modulation order $M$ is smaller than $N$, the inference accuracy is lower than $90\%$ even at high SNR. In particular, for the high SNR, when the modulation order $M$ increases to $N$, the inference accuracy will be significantly improved, and the inference accuracy will not increase with the increase of $M$. Fig. \ref{img5} (a) and  (b) verify that the semantic entropy of the classification task is equal to the Shannon entropy of the category.
		\begin{figure}
	\centering
	\begin{minipage}[b]{0.55\textwidth}
		\centering
		\includegraphics[width=9.5cm]{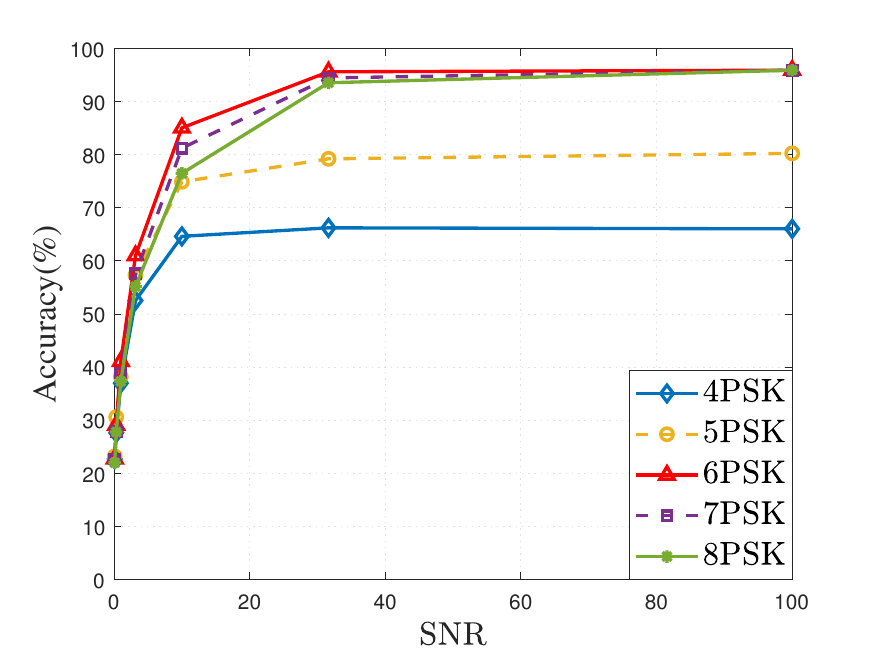}
		\vskip-0.2cm\centering {\footnotesize (a)}
	\end{minipage}\hfill
	\vskip 0.1cm
	\begin{minipage}[b]{0.55\textwidth}
		\centering
		\includegraphics[width=9.5cm]{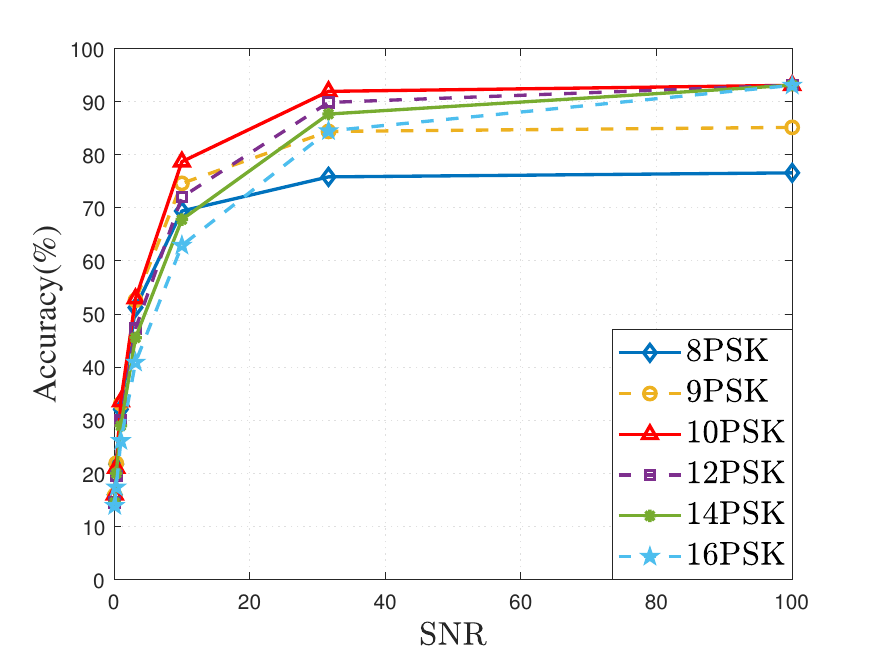}
		\vskip-0.2cm\centering {\footnotesize (b)}
	\end{minipage}\hfill
	\caption{(a)   the inference performance of 6-classification task and (b)   the inference performance of 10-classification task.
	}
	\label{img5}  
\end{figure}
		
		\subsubsection{Semantic channel coding theorem verification}
		Since complex Gaussian noise is utilized in the adopted model, the channel capacity $C = {\log _2}\left( {1 + {\text{SNR}}} \right)$. We already know the average size of each image of the MNIST dataset is about $6272{\text{ bits}}$. It can be seen from Fig. \ref{img5} (a) and  (b) that when the SNR is 63, the channel capacity is 6 ${\text{bit/s}}$, which is far less than the source entropy. However, when the modulation order $M \ge N$, the receiver  can still achieve high inference accuracy.
		In addition,  when the signal rate through the channel exceeds the channel capacity of Shannon theorem, the BER is significantly increased, and the information quality is seriously decreased. In fact, the channel capacity is only the limit that can be reached theoretically, and how to reach it in practice cannot be explained by the theorem.
		In Fig. \ref{img5} (b), we take 10PSK with ${\text{N = 10}}$ as an example, its transmission rate ${\text{R = }}{\log _2}10{\text{ bit/s}}$, when SNR is 9, the channel capacity ${\text{C = }}{\log _2}10{\text{ bit/s}}$, and the classification accuracy is about $75\% $. We also observe that  $R > C$ from Fig. \ref{img5} (a) and  (b),  some images are still accurately classified and the accuracy is greater than the $\frac{1}{{10}}$ probability of blind guess. The above experiments verify the rationality of the proposed semantic channel coding theorem.

		\section{Conclusions}
		In this paper,  we present  a more general definition of  semantic entropy and prove the semantic channel coding theorem for semantic communications. Different from existing definitions of semantic entropy which can only be applied to specific tasks, we  define  the       semantic entropy as the uncertainty in the semantic interpretation of random variable symbols in the context of knowledge bases, which not only depends on the probability distribution, but also depends on the specific value of the symbol and the background knowledge base.
		Moreover, 	semantic entropy can be transformed into existing semantic entropy definitions and Shannon entropy.
		Furthermore, exploiting the fact that different bits can have the same semantic information, we first propose   semantic channel coding theorem, and prove its achievability and converse.
		Our work attempts to provide   fundamentals  for the study of semantic communication.

		\bibliographystyle{IEEEtran}

		\bibliography{ref0216}                        

\begin{thebibliography}{10}
\providecommand{\url}[1]{#1}
\csname url@samestyle\endcsname
\providecommand{\newblock}{\relax}
\providecommand{\bibinfo}[2]{#2}
\providecommand{\BIBentrySTDinterwordspacing}{\spaceskip=0pt\relax}
\providecommand{\BIBentryALTinterwordstretchfactor}{4}
\providecommand{\BIBentryALTinterwordspacing}{\spaceskip=\fontdimen2\font plus
\BIBentryALTinterwordstretchfactor\fontdimen3\font minus
  \fontdimen4\font\relax}
\providecommand{\BIBforeignlanguage}[2]{{%
\expandafter\ifx\csname l@#1\endcsname\relax
\typeout{** WARNING: IEEEtran.bst: No hyphenation pattern has been}%
\typeout{** loaded for the language `#1'. Using the pattern for}%
\typeout{** the default language instead.}%
\else
\language=\csname l@#1\endcsname
\fi
#2}}
\providecommand{\BIBdecl}{\relax}
\BIBdecl

\bibitem{Shannon_ACM_2001}
C.~E. Shannon, ``A mathematical theory of communication,'' \emph{ACM SIGMOBILE
  Mobile Comput. Commun. Rev.}, vol.~5, no.~1, pp. 3--55, 2001.

\bibitem{Anthes_AC_2016}
C.~Anthes, R.~J. Garc{\'\i}a-Hern{\'a}ndez, M.~Wiedemann, and
  D.~Kranzlm{\"u}ller, ``State of the art of virtual reality technology,'' in
  \emph{IEEE Aerosp. conf.}, Jun. 2016, pp. 1--19.

\bibitem{Voulodimos_CIN_2018}
A.~Voulodimos, N.~Doulamis, A.~Doulamis, E.~Protopapadakis \emph{et~al.},
  ``Deep learning for computer vision: A brief review,'' \emph{Comput. Intell.
  Neuroscience}, Feb. 2018.

\bibitem{Akyildiz_JFET_2022}
I.~F. Akyildiz and H.~Guo, ``Holographic-type communication: A new challenge
  for the next decade,'' \emph{ITU J. Future and Evolving Technologies.}, Sept.
  2022.

\bibitem{Niu_MC_2021}
K.~Niu, J.~Dai, and P.~Zhang, ``Semantic communication for 6{G},'' \emph{Mobile
  Commun.}, vol.~45, no.~04, pp. 85--90, Jul. 2021.

\bibitem{Zhang_E_2022}
P.~Zhang, W.~Xu, H.~Gao, K.~Niu, X.~Xu, X.~Qin, C.~Yuan, Z.~Qin, H.~Zhao,
  J.~Wei \emph{et~al.}, ``Toward wisdom-evolutionary and primitive-concise
  6{G}: A new paradigm of semantic communication networks,'' \emph{Eng.},
  vol.~8, pp. 60--73, Nov. 2021.

\bibitem{Shi_CJIT_2022}
G.~Shi, Y.~Xiao, Y.~Li, D.~Gao, and X.~Xie, ``Semantic communication network
  for intelligent connection of all things,'' \emph{Chin. J. Int. Things},
  vol.~5, no.~02, pp. 26--36, Apr. 2021.

\bibitem{Liu_ISIT_2021}
J.~Liu, W.~Zhang, and H.~V. Poor, ``A rate-distortion framework for
  characterizing semantic information,'' in \emph{IEEE Int. Symp. Inf. Theory
  (ISIT)}, Jul. 2021, pp. 2894--2899.

\bibitem{Wang_WC_2020}
C.-X. Wang, M.~D. Renzo, S.~Stanczak, S.~Wang, and E.~G. Larsson, ``Artificial
  intelligence enabled wireless networking for 5{G} and beyond: Recent advances
  and future challenges,'' \emph{IEEE Wirel. Commun.}, vol.~27, no.~1, pp.
  16--23, Feb. 2020.

\bibitem{Hong_SCIC_2016}
W.~Hong, C.~Yu, J.~Chen, and Z.~Hao, ``Millimeter wave and terahertz
  technology,'' \emph{Sci. China: Inf. Sci.}, vol.~46, no.~8, pp. 1086--1107,
  2016.

\bibitem{Du_ITJ_2023}
B.~Du, H.~Du, H.~Liu, D.~Niyato, P.~Xin, J.~Yu, M.~Qi, and Y.~Tang,
  ``Yolo-based semantic communication with generative ai-aided resource
  allocation for digital twins construction,'' \emph{IEEE Internet of Things
  J.}, pp. 1--15, Sept. 2023.

\bibitem{Cang_ITJ_2023}
Y.~Cang, M.~Chen, Z.~Yang, Y.~Hu, Y.~Wang, C.~Huang, and Z.~Zhang, ``Online
  resource allocation for semantic-aware edge computing systems,'' \emph{IEEE
  Internet of Things J.}, pp. 1--16, Oct. 2023.

\bibitem{Huang_ITJ_2023}
J.~Huang, D.~Li, C.~Huang, X.~Qin, and W.~Zhang, ``Joint task and data oriented
  semantic communications: A deep separate source-channel coding scheme,''
  \emph{IEEE Internet Things J.}, pp. 1--18, Jul. 2023.

\bibitem{Zhang_SCIC_2022}
Y.~Zhang, P.~Zhang, Q.~Wei, H.~Zhao, J.~Xiong, and J.~Zhang, ``Semantic
  communication for agents: Architecture and example,'' \emph{Sc. China: Inf.
  Sci.}, vol.~52, no.~05, pp. 907--921, May. 2022.

\bibitem{Xu_IEEE_2023}
W.~Xu, Z.~Yang, D.~W.~K. Ng, M.~Levorato, Y.~C. Eldar, and M.~Debbah, ``Edge
  learning for {B}5{G} networks with distributed signal processing: Semantic
  communication, edge computing, and wireless sensing,'' \emph{IEEE J. Sel.
  Topics Signal Process.}, vol.~17, no.~1, pp. 9--39, Jan. 2023.

\bibitem{Lan_JCIN_2021}
Q.~Lan, D.~Wen, Z.~Zhang, Q.~Zeng, X.~Chen, P.~Popovski, and K.~Huang, ``What
  is semantic communication? {A} view on conveying meaning in the era of
  machine intelligence,'' \emph{J. Commun. Inf. Netw.}, vol.~6, no.~4, pp.
  336--371, Dec. 2021.

\bibitem{Shannon_MG_1950}
C.~E. Shannon and W.~Weaver, ``The mathematical theory of communication,''
  vol.~34, no. 310, 1950, pp. 312--313.

\bibitem{Shi_CM_2021}
G.~Shi, Y.~Xiao, Y.~Li, and X.~Xie, ``From semantic communication to
  semantic-aware networking: Model, architecture, and open problems,''
  \emph{IEEE Commun. Mag.}, vol.~59, no.~8, pp. 44--50, Aug. 2021.

\bibitem{Strinati_CN_2021}
E.~C. Strinati and S.~Barbarossa, ``6{G} networks: Beyond shannon towards
  semantic and goal-oriented communications,'' \emph{Comp. Netw.}, vol. 190, p.
  107930, May. 2021.

\bibitem{Zhu_Arxiv_2023}
L.~Zhonghao, Z.~Guangxu, X.~Jie, A.~Bo, and C.~Shuguang, ``Semantic
  communications for image recovery and classification via deep joint source
  and channel coding,'' Apr. 2023.

\bibitem{Shi'WC23}
Y.~Shi, Y.~Zhou, D.~Wen, Y.~Wu, C.~Jiang, and K.~B. Letaief, ``Task-oriented
  communications for 6g: Vision, principles, and technologies,'' \emph{IEEE
  Wireless Communications}, vol.~30, no.~3, pp. 78--85, 2023.

\bibitem{Xie_JSAC_2021}
H.~Xie and Z.~Qin, ``A lite distributed semantic communication system for
  internet of things,'' \emph{IEEE J. Sel. Areas Commun.}, vol.~39, no.~1, pp.
  142--153, Jan. 2021.

\bibitem{Xie_ITJ_2023}
B.~Xie, Y.~Wu, Y.~Shi, D.~W.~K. Ng, and W.~Zhang, ``Communication-efficient
  framework for distributed image semantic wireless transmission,'' \emph{IEEE
  Internet of Things J.}, pp. 1--15, Aug. 2023.

\bibitem{Shao_JSAC_2022}
J.~Shao, Y.~Mao, and J.~Zhang, ``Learning task-oriented communication for edge
  inference: An information bottleneck approach,'' \emph{IEEE J. Sel. Areas
  Commun.}, vol.~40, no.~1, pp. 197--211, Jan. 2022.

\bibitem{Kang_2022}
X.~Kang, B.~Song, J.~Guo, Z.~Qin, and F.~R. Yu, ``Task-{O}riented image
  transmission for scene classification in unmanned aerial systems,''
  \emph{IEEE Trans. on Commun.}, vol.~70, no.~8, pp. 5181--5192, Jun. 2022.

\bibitem{Qin_arXiv_2021}
Z.~Qin, X.~Tao, J.~Lu, and G.~Y. Li, ``Semantic communications: Principles and
  challenges,'' \emph{arXiv preprint arXiv:2201.01389}, Dec. 2021.

\bibitem{Xin_E_2022}
G.~Xin and P.~Fan, ``{EXK}-{SC}: A semantic communication model based on
  information framework expansion and knowledge collision,'' \emph{Entropy},
  vol.~24, no.~12, p. 1842, Oct. 2022.

\bibitem{Carnap_RLE_1952}
R.~Carnap, Y.~Bar-Hillel \emph{et~al.}, ``An outline of a theory of semantic
  information,'' 1952.

\bibitem{Floridi_MM_2004}
L.~Floridi, ``Outline of a theory of strongly semantic information,''
  \emph{Minds Mach.}, vol.~14, no.~2, pp. 197--221, May. 2004.

\bibitem{Basu_PMC_2014}
P.~Basu, J.~Bao, M.~Dean, and J.~Hendler, ``Preserving quality of information
  by using semantic relationships,'' \emph{IEEE Int. Conf. Pervasive Comput.
  Commun. Workshops}, pp. 58--63, May. 2012.

\bibitem{LUCA_IAC_1972}
A.~DE~LUCA and S.~TERMINI, ``A definition of a nonprobabilistic entropy in the
  setting of fuzzy sets theory,'' \emph{Inf. Control}, vol.~20, pp. 301--312,
  May. 1972.

\bibitem{Liu_ITFS_2020}
X.~Liu, W.~Jia, W.~Liu, and W.~Pedrycz, ``{AFSSE}: An interpretable classifier
  with axiomatic fuzzy set and semantic entropy,'' \emph{IEEE Trans. Fuzzy
  Syst.}, vol.~28, no.~11, pp. 2825--2840, Oct. 2020.

\bibitem{Venhuizen_MDPI_2019}
N.~J. Venhuizen, M.~W. Crocker, and H.~Brouwer, ``Semantic entropy in language
  comprehension,'' \emph{Entropy}, vol.~21, no.~12, p. 1159, Nov. 2019.

\bibitem{Chattopadhyay_2020}
A.~Chattopadhyay, B.~D. Haeffele, D.~Geman, and R.~Vidal, ``Quantifying task
  complexity through generalized information measures,'' Sept. 2020.

\bibitem{Kountouris_CM_2021}
M.~Kountouris and N.~Pappas, ``Semantics-{E}mpowered communication for
  networked intelligent systems,'' \emph{IEEE Commun. Mag.}, vol.~59, no.~6,
  pp. 96--102, Jan. 2021.

\bibitem{Renyi_1961}
A.~R{\'e}nyi, ``On measures of entropy and information,'' in \emph{Proc.
  Berkeley Symp. Math. Statist. Probability}, vol.~4, 1961, pp. 547--562.

\bibitem{Kolchinsky_IF_2018}
A.~Kolchinsky and D.~H. Wolpert, ``Semantic information, autonomous agency and
  non-equilibrium statistical physics,'' \emph{Interface Focus}, vol.~8, no.~6,
  p. 20180041, 2018.

\bibitem{6004632}
J.~Bao, P.~Basu, M.~Dean, C.~Partridge, A.~Swami, W.~Leland, and J.~A. Hendler,
  ``Towards a theory of semantic communication,'' in \emph{IEEE Netw. Sci.
  Workshop}, Jun. 2011, pp. 110--117.

\bibitem{Cover_book_1999}
T.~M. Cover, \emph{Elements of information theory}.\hskip 1em plus 0.5em minus
  0.4em\relax John Wiley \& Sons, 1999.

\bibitem{Basak}
B.~G{\"u}ler, A.~Yener, and A.~Swami, ``The semantic communication game,'' in
  \emph{IEEE Int. Conf. Commun.}, May. 2016, pp. 1--6.

\bibitem{Jinho_2022}
J.~Choi and J.~Park, ``Semantic communication as a signaling game with
  correlated knowledge bases,'' in \emph{IEEE Veh. Technol. Conf.
  (VTC2022-Fall)}, Jan. 2022, pp. 1--5.

\bibitem{Xie_2023}
S.~Xie, S.~Ma, M.~Ding, Y.~Shi, M.~Tang, and Y.~Wu, ``Robust information
  bottleneck for task-oriented communication with digital modulation,''
  \emph{IEEE J. Sel. Areas Commun.}, pp. 1--15, Jun. 2023.

\end{thebibliography}
		
	\end{document}